\documentclass[conference]{IEEEtran}%
\pdfoutput=1
\usepackage{amsfonts}
\usepackage{amsmath}
\usepackage{amssymb}
\usepackage{graphicx}
\IEEEoverridecommandlockouts
\providecommand{\U}[1]{\protect\rule{.1in}{.1in}}
\newtheorem{theorem}{Theorem}

\newtheorem{corollary}[theorem]{Corollary}

\newtheorem{lemma}[theorem]{Lemma}

\newtheorem{remark}[theorem]{Remark}

\newenvironment{proof}[1][Proof]{\noindent\textbf{#1.} }{\ \rule{0.5em}{0.5em}}
\begin{document}

\title{Explicit receivers for pure-interference bosonic multiple access channels}
\author{\IEEEauthorblockN{Mark M. Wilde} \IEEEauthorblockA{\it School of Computer Science\\McGill University\\Montreal, Qu\'{e}bec,
Canada H3A 2A7} \and
\IEEEauthorblockN{Saikat Guha\thanks{\tiny{MMW acknowledges financial support from Centre de Recherches Math\'{e}matiques. SG was supported by the DARPA Information in a Photon (InPho) program under DARPA/CMO Contract No. HR0011-10-C-0159. The views and conclusions contained in this document are those of the authors and should not be interpreted as representing the official policies, either expressly or implied, of the Defense Advanced Research Projects Agency or the U.S. Government.}}}
\IEEEauthorblockA{\it Disruptive Information Proc. Tech. Group\\Raytheon BBN
Technologies\\Cambridge, Massachusetts, USA 02138}}

\maketitle

\begin{abstract}
The pure-interference bosonic multiple access channel has two senders and one receiver,
such that the senders each communicate with multiple temporal modes of a single spatial
mode of light. The channel mixes the input modes from the two users pairwise on a lossless
beamsplitter, and the receiver has access to one of the two output ports.
In prior work, Yen and Shapiro found the capacity region of this channel if
encodings consist of coherent-state preparations.
Here, we demonstrate how to achieve the coherent-state Yen-Shapiro region
(for a range of parameters) using
a sequential decoding strategy, and we show that our strategy outperforms
the rate regions achievable
using conventional receivers. Our
receiver performs binary-outcome quantum measurements for every codeword pair
in the senders' codebooks. A crucial component of this scheme is a
non-destructive \textquotedblleft vacuum-or-not\textquotedblright measurement
that projects an $n$-symbol modulated codeword onto the $n$-fold vacuum
state or its orthogonal complement, such that the post-measurement state
is either the $n$-fold vacuum or has the vacuum removed from the support
of the $n$ symbols' joint quantum state. This receiver requires the additional
ability to perform multimode optical phase-space displacements which are
realizable using a beamsplitter and a laser.

\end{abstract}

One of the most important questions in quantum information theory is to
determine the maximum rate at which it is possible to transmit data error-free
over many independent uses of a noisy quantum channel. In the spirit of
Shannon~\cite{bell1948shannon}, this quantity is known as the
\textit{classical} capacity of a quantum channel because it has to do with the
transmission of bits, \textquotedblleft classical\textquotedblright\ data,
over a quantum channel. Holevo, Schumacher, and Westmoreland (HSW) made
partial progress on this question by providing a good lower bound on any
channel's classical capacity \cite{H98,PhysRevA.56.131}. For many channels,
the HSW\ lower bound is equal to the capacity, but in general, it is not
\cite{H09}.

A channel for which the HSW\ lower bound is equal to the classical capacity is
the pure-loss bosonic channel \cite{GGLMSY04}. This capacity result follows
because the HSW lower bound coincides with the Yuen-Ozawa upper bound on the
channel's classical capacity \cite{YO93}. The pure-loss bosonic channel is a
reasonable model of free-space or fiber-optic communication \cite{S09} and has
the following Heisenberg-picture specification:%
\begin{equation}
\hat{b}=\sqrt{\eta}\ \hat{a}+\sqrt{1-\eta}\ \hat{e}, \label{eq:pure-loss-bos}%
\end{equation}
where $\hat{a}$, $\hat{b}$, and $\hat{e}$ are the EM\ field mode operators for
the sender, receiver, and environment, respectively, and $\eta\in\left[
0,1\right]  $ is a transmissivity parameter determining what fraction of
photons make it to the receiver on average. An assumption for this
channel is that the environmental input is the vacuum state. Another realistic
assumption usually made when calculating this channel's classical capacity is
that the sender is constrained to have a finite mean photon-number budget
of~$N_{S}$ (otherwise, the channel's classical capacity is infinite). In this
case, the channel's classical capacity is equal to $g\left(  \eta
N_{S}\right)  $, where%
\[
g\left(  x\right)  \equiv\left(  x+1\right)  \log_{2}\left(  x+1\right)
-x\log_{2}x.
\]

The authors of Ref.~\cite{GGLMSY04} proved that an encoding strategy for
achieving the pure-loss channel's classical capacity is to generate
tensor-product coherent-state codewords randomly according to a complex,
isotropic Gaussian distribution with variance $N_{S}$ (this allows the
strategy to meet the mean photon budget constraint of $N_{S}$). After doing
so, the codebook consists of coherent-state codewords of the form:%
\begin{equation}
\left\vert \alpha^{n}\left(  m\right)  \right\rangle \equiv\left\vert
\alpha_{1}\left(  m\right)  \right\rangle \otimes\left\vert \alpha_{2}\left(
m\right)  \right\rangle \otimes\cdots\otimes\left\vert \alpha_{n}\left(
m\right)  \right\rangle , \label{eq:coh-state-codewords}%
\end{equation}
where $m\in \mathcal{M}$ indicates the classical message to be sent and $\alpha_{1}\left(
m\right)  $, \ldots, $\alpha_{n}\left(  m\right)  \in\mathbb{C}$ (these are
the independent realizations of the complex Gaussian random variable with
variance $N_{S}$, where $\left\vert \alpha_{i}\left(  m\right)  \right\rangle
$ is a coherent state). Recently, we showed that a sequential decoding
strategy consisting of binary-outcome quantum measurements of the form%
\begin{equation}
\left\{  \left\vert \alpha^{n}\left(  m\right)  \right\rangle \left\langle
\alpha^{n}\left(  m\right)  \right\vert ,I^{\otimes n}-\left\vert \alpha
^{n}\left(  m\right)  \right\rangle \left\langle \alpha^{n}\left(  m\right)
\right\vert \right\}  , \label{eq:bin-meas-single-send}%
\end{equation}
suffices to achieve the classical capacity of this channel \cite{GTW12}. This
result gives an explicit, physical way to realize a capacity-achieving
receiver in terms of linear-optical displacements and a coherent,
non-demolition \textquotedblleft vacuum-or-not\textquotedblright\ measurement.

A natural multi-user extension of a single-sender, single-receiver channel is one with two senders and
one receiver (a multiple-access channel). Determining strategies for
communication over such channels will be important for multi-user
communication in free space. Yen and Shapiro considered a simple multi-user
extension of the pure-loss bosonic channel in (\ref{eq:pure-loss-bos}), and
they called it the pure-interference multiple access channel
(MAC)\ \cite{YS05}. It has the following input-output Heisenberg-picture
specification:%
\begin{equation}
\hat{c}=\sqrt{\eta}\ \hat{a}+\sqrt{1-\eta}\ \hat{b}, \label{eq:bosonic-MAC}%
\end{equation}
where $\hat{a}$, $\hat{b}$, and $\hat{c}$ are the EM\ field mode operators for
the first sender, the second sender, and the receiver, respectively, and
$\eta\in\left[  0,1\right]  $ is an interference parameter determining how
much the senders' transmissions mix. In this model, the only noise that occurs
is due to the mixing of the senders' transmissions. Also, we allow the first
sender a mean photon budget of $N_{S_{A}}$ and the second sender a budget of
$N_{S_{B}}$.

Yen and Shapiro called the above channel the \textquotedblleft coherent-state
MAC\textquotedblright\ if the senders are restricted to using only
coherent-state input codewords, and they proved the capacity region in this
case has the following form:%
\begin{align}
R_{1}  &  \leq g\left(  \eta N_{S_{A}}\right)  ,\nonumber\\
R_{2}  &  \leq g\left(  \left(  1-\eta\right)  N_{S_{B}}\right)  ,\nonumber\\
R_{1}+R_{2}  &  \leq g\left(  \eta N_{S_{A}}+\left(  1-\eta\right)  N_{S_{B}%
}\right)  , \label{eq:YS-region}%
\end{align}
where $R_{1}$ and $R_{2}$ are the first and second senders' communication
rates, respectively. They proved this result by invoking Winter's theorem for
coding over a general quantum multiple access channel \cite{W01}, and they
showed that, in principle, a quantum strategy at the receiver can outperform a
conventional classical strategy such as homodyne or heterodyne detection. The
first sender chooses a coherent-state codebook of the form $\left\{
\left\vert \alpha^{n}\left(  l\right)  \right\rangle \right\}  _{l}$, and the
second sender similarly chooses a codebook of the form $\left\{  \left\vert
\beta^{n}\left(  m\right)  \right\rangle \right\}  _{m}$, with the codewords
defined similarly as in (\ref{eq:coh-state-codewords}). If the first sender
chooses message~$l \in \mathcal{L}$ and the second sender chooses
message~$m\in \mathcal{M}$, then the state
produced at the output of the channel is of the form:%
\begin{equation}
\left\vert \gamma^{n}\left(  l,m\right)  \right\rangle \equiv\left\vert
\gamma_{1}\left(  l,m\right)  \right\rangle \otimes\left\vert \gamma
_{2}\left(  l,m\right)  \right\rangle \otimes\cdots\otimes\left\vert
\gamma_{n}\left(  l,m\right)  \right\rangle , \label{eq:output-states}%
\end{equation}
where%
\begin{equation}
\gamma_{i}\left(  l,m\right)  \equiv\sqrt{\eta}\ \alpha_{i}\left(  l\right)
+\sqrt{1-\eta}\ \beta_{i}\left(  m\right)  . \label{eq:gam-mix-def}%
\end{equation}
In spite of finding the above physical realization for the encoder, Yen and Shapiro
left open the question of giving a physically-realizable form for a quantum
receiver that achieves the above rate region.\footnote{Since their result
relies on Winter's \cite{W01}, the collective measurement at the receiving end
is a \textquotedblleft square-root\textquotedblright\ measurement. This
measurement is well-known within the quantum information theory community, but
it is unclear how one might implement it with optical devices.}

In this paper, we prove that in some cases a sequential decoding strategy,
realized explicitly with optical devices (and a multimode non-destructive vacuum or not
measurement for which we do not know a structured optical realization yet), can achieve the rate region in
(\ref{eq:YS-region}) for the pure-interference bosonic multiple access
channel. This sequential decoder is a natural extension of the strategy in
(\ref{eq:bin-meas-single-send}) for the single-sender channel. The receiver
sequentially tests pairs of codewords, by performing binary-outcome quantum
measurements of the following form:%
\begin{equation}
\left\{  \left\vert \gamma^{n}\left(  l,m\right)  \right\rangle \left\langle
\gamma^{n}\left(  l,m\right)  \right\vert ,I^{\otimes n}-\left\vert \gamma
^{n}\left(  l,m\right)  \right\rangle \left\langle \gamma^{n}\left(
l,m\right)  \right\vert \right\}  . \label{eq:pair-decoder}%
\end{equation}
The cases for which this sequential decoding strategy achieves the rates in
(\ref{eq:YS-region}) have to do with the mean-photon number constraints
$N_{S_{A}}$ and $N_{S_{B}}$ and the transmissivity $\eta$, and we later outline specifically for which values of
these parameters the decoder in (\ref{eq:pair-decoder}) can achieve the rate
region in (\ref{eq:YS-region}).

We structure this paper as follows. In the next section, we overview some basic definitions
that we use throughout the paper. In Section~\ref{sec:main-thm},
we state our main theorem regarding sequential decoding
for a general pure-state multiple access channel with two classical inputs and one pure-state
quantum output. This section outlines a proof of this theorem with the bulk of it appearing
in Appendix~\ref{sec:err-analysis}. In Section~\ref{sec:rec-alg-bos-MAC}, we apply
this theorem to the pure-interference bosonic multiple access channel. Section~\ref{sec:examples} discusses some
cases for which a rate region achievable with sequential decoding is
equivalent to the Yen-Shapiro region in (\ref{eq:YS-region}). Finally, we conclude in Section~\ref{sec:conc} with a
summary and some open questions.

\section{Notation and Definitions}

We denote pure states of a quantum
 system $A$ with a \emph{ket} $\left\vert \phi\right\rangle ^{A}$ and the
corresponding density operator as $\phi^{A}=\left\vert \phi\right\rangle
\!\left\langle \phi\right\vert ^{A}$. All kets that are quantum states have
unit norm, and all density operators are positive semi-definite with unit
trace. Let
$
H(A)_{\rho}\equiv-\text{Tr}\left\{  \rho^{A}\log_2\rho^{A}\right\}
$
be the von Neumann entropy of the state $\rho^{A}$. For a state $\sigma^{ABC}%
$, we define the quantum conditional entropy
$
H(A|B)_{\sigma}\equiv H(AB)_{\sigma}-H(B)_{\sigma}
$
and the quantum mutual information
$
I(A;B)_{\sigma}\equiv H(A)_{\sigma}+H(B)_{\sigma}-H(AB)_{\sigma}.
$
In order to describe the \textquotedblleft distance\textquotedblright\ between
two quantum states, we use the notion of \emph{trace distance}. The trace
distance between states $\sigma$ and $\rho$ is
\[
\Vert\sigma-\rho\Vert_{1}=\mathrm{Tr}\left\vert \sigma-\rho\right\vert ,
\]
where $|X|=\sqrt{X^{\dagger}X}$. Two states that are similar have trace
distance close to zero, whereas states that are perfectly distinguishable have
trace distance equal to two.

The min-entropy $H_{\min}\left(  B\right)  _{\rho}$ of a quantum state
$\rho^{B}$ is equal to the negative logarithm of its maximal eigenvalue:%
\[
H_{\min}\left(  B\right)  _{\rho}\equiv-\log_2\left(  \inf_{\lambda\in
\mathbb{R}}\left\{  \lambda:\rho\leq\lambda I\right\}  \right)  ,
\]
and the conditional min-entropy of a classical-quantum state $\rho^{XB}%
\equiv\sum_{x}p_{X}\left(  x\right)  \left\vert x\right\rangle \left\langle
x\right\vert ^{X}\otimes\rho_{x}^{B}$ with classical system $X$ and quantum
system $B$ is as follows:%
\[
H_{\min}\left(  B|X\right)  _{\rho}\equiv\inf_{x\in\mathcal{X}}H_{\min}\left(
B\right)  _{\rho_{x}}.
\]
This definition of conditional min-entropy, where the conditioning system is
classical, implies the following operator inequality:%
\begin{equation}
\forall x\ \ \ \rho_{x}^{B} \leq2^{-H_{\min}\left(  B|X\right)  _{\rho}}I^{B}.
\label{eq:cond-min-entropy-bound}%
\end{equation}

\section{General scheme for a pure-state output multiple access channel}
\label{sec:main-thm}
We now prove a general result regarding rates that are achievable over a
pure-state multiple access channel of the following form:%
\begin{equation}
x,y\rightarrow\left\vert \phi_{x,y}\right\rangle , \label{eq:pure-state-MAC}%
\end{equation}
such that Sender~1 input the letter $x$, Sender~2 inputs the letter $y$, and
the receiver obtains the quantum state $\left\vert \phi_{x,y}\right\rangle $
at the output of the channel.

\begin{theorem}
\label{thm:main-theorem}Suppose that the receiver of the pure-state multiple
access channel in (\ref{eq:pure-state-MAC}) is restricted to using a
sequential decoder with binary-outcome tests of the form $\left\{  \left\vert
\phi_{x^{n},y^{n}}\right\rangle \left\langle \phi_{x^{n},y^{n}}\right\vert
,I-\left\vert \phi_{x^{n},y^{n}}\right\rangle \left\langle \phi_{x^{n},y^{n}%
}\right\vert \right\}  $. Then the following rate region is achievable for
communication over this channel:%
\begin{align}
R_{1}  &  \leq H_{\min}\left(  B|Y\right)  ,\nonumber\\
R_{2}  &  \leq H\left(  B|X\right)  ,\nonumber\\
R_{1}+R_{2}  &  \leq H\left(  B\right)  , \label{eq:first-region}%
\end{align}
where the entropies are with respect to a classical-quantum state of the
following form, for some distributions $p_{X}\left(  x\right)  $ and
$p_{Y}\left(  y\right)  $:%
\[
\sum_{x,y}p_{X}\left(  x\right)  p_{Y}\left(  y\right)  \left\vert
x\right\rangle \left\langle x\right\vert ^{X}\otimes\left\vert y\right\rangle
\left\langle y\right\vert ^{Y}\otimes\left\vert \phi_{x,y}\right\rangle
\left\langle \phi_{x,y}\right\vert ^{B}.
\]

\end{theorem}

\begin{proof}
We break the proof into several parts:\ codebook construction, a discussion of
the sequential decoder, and a detailed error analysis appearing in Appendix~\ref{sec:err-analysis}.

\textbf{Codebook Construction.} Before communication begins, Sender~1,
Sender~2, and the receiver agree upon a codebook. We allow Sender~1 to select
a codebook randomly according to the distribution $p_{X}\left(  x\right)  $,
and Sender~2 likewise to select one according to $p_{Y}\left(  y\right)  $.
So, for every message $l\in\mathcal{L}\equiv\left\{  1,\ldots,2^{nR_{1}%
}\right\}  $, generate a codeword $x^{n}\left(  l\right)  \equiv x_{1}\left(
l\right)  \cdots x_{n}\left(  l\right)  $ randomly and independently according
to%
\[
p_{X^{n}}\left(  x^{n}\right)  \equiv\prod\limits_{i=1}^{n}p_{X}\left(
x_{i}\right)  .
\]
Similarly, for every message $m\in\mathcal{M}\equiv\left\{  1,\ldots
,2^{nR_{2}}\right\}  $, generate a codeword $y^{n}\left(  m\right)  \equiv
y_{1}\left(  m\right)  \cdots y_{n}\left(  m\right)  $ randomly and
independently according to%
\[
p_{Y^{n}}\left(  y^{n}\right)  \equiv\prod\limits_{i=1}^{n}p_{Y}\left(
y_{i}\right)  .
\]

\textbf{Sequential Decoding.} Transmitting the codewords $x^{n}\left(
l\right)  $ and $y^{n}\left(  m\right)  $ through $n$ uses of the channel
$x,y\rightarrow\left\vert \phi_{x,y}\right\rangle $ leads to the following
quantum state at the receiver's output:%
\[
\left\vert \phi_{x^{n}\left(  l\right)  ,y^{n}\left(  m\right)  }\right\rangle
\equiv\left\vert \phi_{x_{1}\left(  l\right)  ,y_{1}\left(  m\right)
}\right\rangle \otimes\cdots\otimes\left\vert \phi_{x_{n}\left(  l\right)
,y_{n}\left(  m\right)  }\right\rangle .
\]
Upon receiving the quantum codeword $\left\vert \phi_{x^{n}\left(  l\right)
,y^{n}\left(  m\right)  }\right\rangle $, the receiver performs a sequence of
binary-outcome quantum measurements to determine the classical codewords
$x^{n}\left(  l\right)  $ and $y^{n}\left(  m\right)  $\ that the senders
transmitted. He first \textquotedblleft asks,\textquotedblright%
\ \textquotedblleft Is it the first codeword pair?\textquotedblright\ by
performing the measurement%
\[
\{\phi_{x^{n}\left(  1\right)  ,y^{n}\left(  1\right)  },I^{\otimes n}%
-\phi_{x^{n}\left(  1\right)  ,y^{n}\left(  1\right)  }\},
\]
where we abbreviate%
\[
\phi_{x^{n}\left(  1\right)  ,y^{n}\left(  1\right)  }\equiv\left\vert
\phi_{x^{n}\left(  1\right)  ,y^{n}\left(  1\right)  }\right\rangle
\left\langle \phi_{x^{n}\left(  1\right)  ,y^{n}\left(  1\right)  }\right\vert
.
\]
If he receives the outcome \textquotedblleft yes,\textquotedblright\ then he
performs no further measurements and concludes that the senders transmitted
the codewords $x^{n}\left(  1\right)  $ and $y^{n}\left(  1\right)  $. If he
receives the outcome \textquotedblleft no,\textquotedblright\ then he performs
the measurement
\[
\{\phi_{x^{n}\left(  2\right)  ,y^{n}\left(  1\right)  },I^{\otimes n}%
-\phi_{x^{n}\left(  2\right)  ,y^{n}\left(  1\right)  }\}.
\]
to check if the senders transmitted the second codeword pair. Similarly, he
stops if he receives \textquotedblleft yes,\textquotedblright\ and otherwise,
he proceeds along similar lines. The order in which the receiver scans through
the codeword pairs is%
\begin{align*}
\left(  x^{n}\left(  1\right)  ,y^{n}\left(  1\right)  \right)   &
\rightarrow\cdots\rightarrow\left(  x^{n}\left(  \left\vert \mathcal{L}%
\right\vert \right)  ,y^{n}\left(  1\right)  \right)  \rightarrow\\
\left(  x^{n}\left(  1\right)  ,y^{n}\left(  2\right)  \right)   &
\rightarrow\cdots\rightarrow\left(  x^{n}\left(  \left\vert \mathcal{L}%
\right\vert \right)  ,y^{n}\left(  2\right)  \right)  \rightarrow\\
&  \cdots\\
\left(  x^{n}\left(  1\right)  ,y^{n}\left(  m\right)  \right)   &
\rightarrow\cdots\rightarrow\left(  x^{n}\left(  l\right)  ,y^{n}\left(
m\right)  \right)  .
\end{align*}

The rest of proof is a detailed error analysis to show that the above scheme
works well. It proceeds similarly to Sen's proof for sequential
decoding of the quantum MAC \cite{S11}, and as such, we put it in
Appendix~\ref{sec:err-analysis}. Though, the difference between our error analysis and Sen's
is that we would like to employ a sequential decoder with measurements of the form
$\left\{  \left\vert
\phi_{x^{n},y^{n}}\right\rangle \left\langle \phi_{x^{n},y^{n}}\right\vert
,I-\left\vert \phi_{x^{n},y^{n}}\right\rangle \left\langle \phi_{x^{n},y^{n}%
}\right\vert \right\}  $, as opposed to the modified measurements that Sen employs
in his proof. This will allow us to have a physically-realizable
decoder for the pure-interference bosonic MAC. We clarify this point in Remark~\ref{rem:connection-to-Sen} of Appendix~\ref{sec:err-analysis}.

\end{proof}

\begin{corollary}
\label{rem:full-region-convex-hull}The following rate region is also
achievable, by considering a symmetric proof in which we smooth the channel to
be $\Pi\Pi_{M}\phi_{L,M}\Pi_{M}\Pi$ rather than $\Pi\Pi_{L}\phi_{L,M}\Pi
_{L}\Pi$ (see Appendix~\ref{sec:err-analysis}). After performing a symmetric error analysis, the resulting rate
region has the following form:%
\begin{align}
R_{1}  &  \leq H\left(  B|Y\right)  ,\,\,\,\,\,
R_{2}  \leq H_{\min}\left(  B|X\right)  ,\nonumber\\
R_{1}+R_{2}  &  \leq H\left(  B\right)  . \label{eq:second-region}%
\end{align}
The convex hull of the above region and the region from
Theorem~\ref{thm:main-theorem} is always achievable because it
corresponds to taking convex combinations of rate pairs from each
rate region and this amounts to a time-sharing strategy. This convex hull region is equivalent to the rate
region $R_1 \leq H(B|Y), R_2 \leq H(B|X), R_1 + R_2 \leq H(B)$ if the corner point%
\[
\left(  R_{1}=\min\left\{  H_{\min}\left(  B|Y\right)  ,I\left(  X;B\right)
\right\}  ,R_{2}=H\left(  B|X\right)  \right)
\]
of the region in (\ref{eq:first-region}) is equal to%
\[
\left(  R_{1}=I\left(  X;B\right)  ,R_{2}=H\left(  B|X\right)  \right)  ,
\]
and if the corner point%
\[
\left(  R_{1}=H\left(  B|Y\right)  ,R_{2}=\min\left\{  H_{\min}\left(
B|X\right)  ,I\left(  Y;B\right)  \right\}  \right)
\]
of the region in (\ref{eq:second-region}) is equal to%
\[
\left(  R_{1}=H\left(  B|Y\right)  ,R_{2}=I\left(  Y;B\right)  \right)  .
\]
We consider these conditions in Section~\ref{sec:rec-alg-bos-MAC}\ when we
reason about the rate regions achievable with sequential decoding for the
pure-interference bosonic multiple access channel.
\end{corollary}

\section{Explicit receiver for the pure-interference bosonic multiple access
channel}

\label{sec:rec-alg-bos-MAC}The strategy for achieving the capacity of the
coherent-state MAC is for the two senders to induce a channel of the form in
(\ref{eq:pure-state-MAC}), by selecting $\alpha,\beta\in{\mathbb{C}}$ and
preparing coherent states $\left\vert \alpha\right\rangle $ and $\left\vert
\beta\right\rangle $ at the input of the channel in (\ref{eq:bosonic-MAC}).
The resulting induced channel to the receiver is of the following form:%
\[
\alpha,\beta\rightarrow|\sqrt{\eta}\alpha+\sqrt{1-\eta}\beta\rangle.
\]
By choosing the distributions $p_{X}\left(  x\right)  $ and $p_{Y}\left(
y\right)  $ in Theorem~\ref{thm:main-theorem}\ to be Gaussian as follows:%
\begin{align*}
p_{N_{S_{A}}}\left(  \alpha\right)   &  \equiv\left(  {1}/{\pi}N_{S_{A}%
}\right)  \exp\left\{  {-\left\vert \alpha\right\vert ^{2}}/N_{S_{A}}\right\}
,\\
p_{N_{S_{B}}}\left(  \beta\right)   &  \equiv\left(  {1}/{\pi}N_{S_{B}%
}\right)  \exp\left\{  {-\left\vert \beta\right\vert ^{2}}/N_{S_{B}}\right\}
,
\end{align*}
we have that the convex hull of the regions in (\ref{eq:first-region}) and
(\ref{eq:second-region}) is achievable (see
Corollary~\ref{rem:full-region-convex-hull}). For our case, the various
entropies become as follows:%
\begin{align*}
H\left(  B|X\right)   &  =g\left(  \left(  1-\eta\right)  N_{S_{B}}\right)
,\\
H\left(  B|Y\right)   &  =g\left(  \eta N_{S_{A}}\right)  ,\\
H_{\min}\left(  B|X\right)   &  =\log_{2}\left(  \left(  1-\eta\right)
N_{S_{B}}+1\right)  ,\\
H_{\min}\left(  B|Y\right)   &  =\log_{2}\left(  \eta N_{S_{A}}+1\right)  ,\\
H\left(  B\right)   &  =g\left(  \eta N_{S_{A}}+\left(  1-\eta\right)
N_{S_{B}}\right)  .
\end{align*}
According to Corollary~\ref{rem:full-region-convex-hull}, we can show that
this strategy achieves the full Yen-Shapiro region in (\ref{eq:YS-region}) if
the both of the following conditions hold%
\begin{align}
g\left(  N^{\prime}\right)  -g\left(  \left(  1-\eta\right)  N_{S_{B}}\right)
&  \leq\log_{2}\left(  \eta N_{S_{A}}+1\right) , \;{\text{and}}  \nonumber\\
g\left(  N^{\prime}\right)  -g\left(  \eta N_{S_{A}}\right)   &  \leq\log
_{2}\left(  \left(  1-\eta\right)  N_{S_{B}}+1\right)  ,
\label{eq:equal-regions-condition}%
\end{align}
where $N^{\prime}\equiv\eta N_{S_{A}}+\left(  1-\eta\right)  N_{S_{B}}$.
Otherwise, the convex hull region from
Corollary~\ref{rem:full-region-convex-hull} is contained in the Yen-Shapiro region.

The quantum codebooks selected from the ensembles $\{p_{N_{S_{A}}}\left(
\alpha\right)  ,\ |\alpha\rangle\}$ and $\{p_{N_{S_{B}}}\left(  \beta\right)
,\ |\beta\rangle\}$ have the respective forms $\left\{  \left\vert \alpha
^{n}\left(  l\right)  \right\rangle \right\}  _{l}$ and $\left\{  \left\vert
\beta^{n}\left(  m\right)  \right\rangle \right\}  _{m}$, with the codewords
defined similarly as in (\ref{eq:coh-state-codewords}). The sequential decoder
consists of binary measurements formed from the output states in
(\ref{eq:output-states}) for all $l\in\mathcal{L}$, $m\in\mathcal{M}$:%
\begin{equation}
\left\{  \left\vert \gamma^{n}\left(  l,m\right)  \right\rangle \left\langle
\gamma^{n}\left(  l,m\right)  \right\vert ,I^{\otimes n}-\left\vert \gamma
^{n}\left(  l,m\right)  \right\rangle \left\langle \gamma^{n}\left(
l,m\right)  \right\vert \right\}  . \label{eq:seq-decoder-bosonic}%
\end{equation}
Observing that%
\[
\left\vert \gamma^{n}\left(  l,m\right)  \right\rangle =D\left(  \gamma
_{1}\left(  l,m\right)  \right)  \otimes\cdots\otimes D\left(  \gamma
_{n}\left(  l,m\right)  \right)  \left\vert 0\right\rangle ^{\otimes n},
\]
where $\gamma_{i}\left(  l,m\right)  $ is defined in (\ref{eq:gam-mix-def}),
 $D\left(  \alpha\right)  \equiv\exp\left\{  \alpha\hat{a}^{\dag}%
-\alpha^{\ast}\hat{a}\right\}  $ is the well-known unitary \textquotedblleft
displacement\textquotedblright\ operator from quantum optics \cite{P96}, and
$\left\vert 0\right\rangle ^{\otimes n}$ is the $n$-fold tensor product vacuum
state, it is clear that the decoder can implement the measurement in
(\ref{eq:seq-decoder-bosonic}) in three steps:

\begin{enumerate}
\item Displace the $n$-mode codeword state by%
\[
D\left(  -\gamma_{1}\left(  l,m\right)  \right)  \otimes\cdots\otimes D\left(
-\gamma_{n}\left(  l,m\right)  \right)  ,
\]
by employing highly asymmetric beam-splitters with a strong local oscillator
\cite{P96}.

\item Perform a non-destructive \textquotedblleft vacuum-or-not\textquotedblright\ measurement
of the form%
\[
\left\{  \left\vert 0\right\rangle \left\langle 0\right\vert ^{\otimes
n},\ \ I^{\otimes n}-\left\vert 0\right\rangle \left\langle 0\right\vert
^{\otimes n}\right\}  .
\]
If the vacuum outcome occurs, decode as the codeword pair $\left(  l,m\right)
$. Otherwise, proceed.

\item Displace by $D\left(  \gamma_{1}\left(  l,m\right)  \right)
\otimes\cdots\otimes D\left(  \gamma_{n}\left(  l,m\right)  \right)  $ with
the same method as in Step 1.
\end{enumerate}

The receiver just iterates this strategy for every codeword pair in the
codebooks.

\section{Examples}
\label{sec:examples}
This section discusses a few examples such that the convex hull region from
Corollary~\ref{rem:full-region-convex-hull} is either equal to the full
Yen-Shapiro rate region in (\ref{eq:YS-region}) or not.

Our first example appears in Figure~\ref{fig:regions}(a). By setting
$\eta=1/2$, $N_{S_{A}}=1$, and $N_{S_{B}}=1$, we find that the conditions in
(\ref{eq:equal-regions-condition}) do not hold, so that the convex hull region
from Corollary~\ref{rem:full-region-convex-hull} is not equal to the full
Yen-Shapiro region. Though, the convex hull region is nearly equal to the
Yen-Shapiro region, and it is significantly larger than the region given by a
classical strategy such as homodyne or heterodyne detection (see
Ref.~\cite{YS05}\ for a discussion of the rate regions resulting from these strategies).

Our second example appears in Figure~\ref{fig:regions}(b), where we set
$\eta=1/2$, $N_{S_{A}}=10$, and $N_{S_{B}}=8$. We find that the conditions in
(\ref{eq:equal-regions-condition}) do hold for these values, implying that the
convex hull region from Corollary~\ref{rem:full-region-convex-hull} is equal
to the full Yen-Shapiro region. Again, this region is significantly larger
than the region given by a classical strategy such as homodyne or heterodyne detection.%

\begin{figure}
[ptb]
\begin{center}
\includegraphics[
natheight=3.620100in,
natwidth=8.146500in,
width=3.7402in
]%
{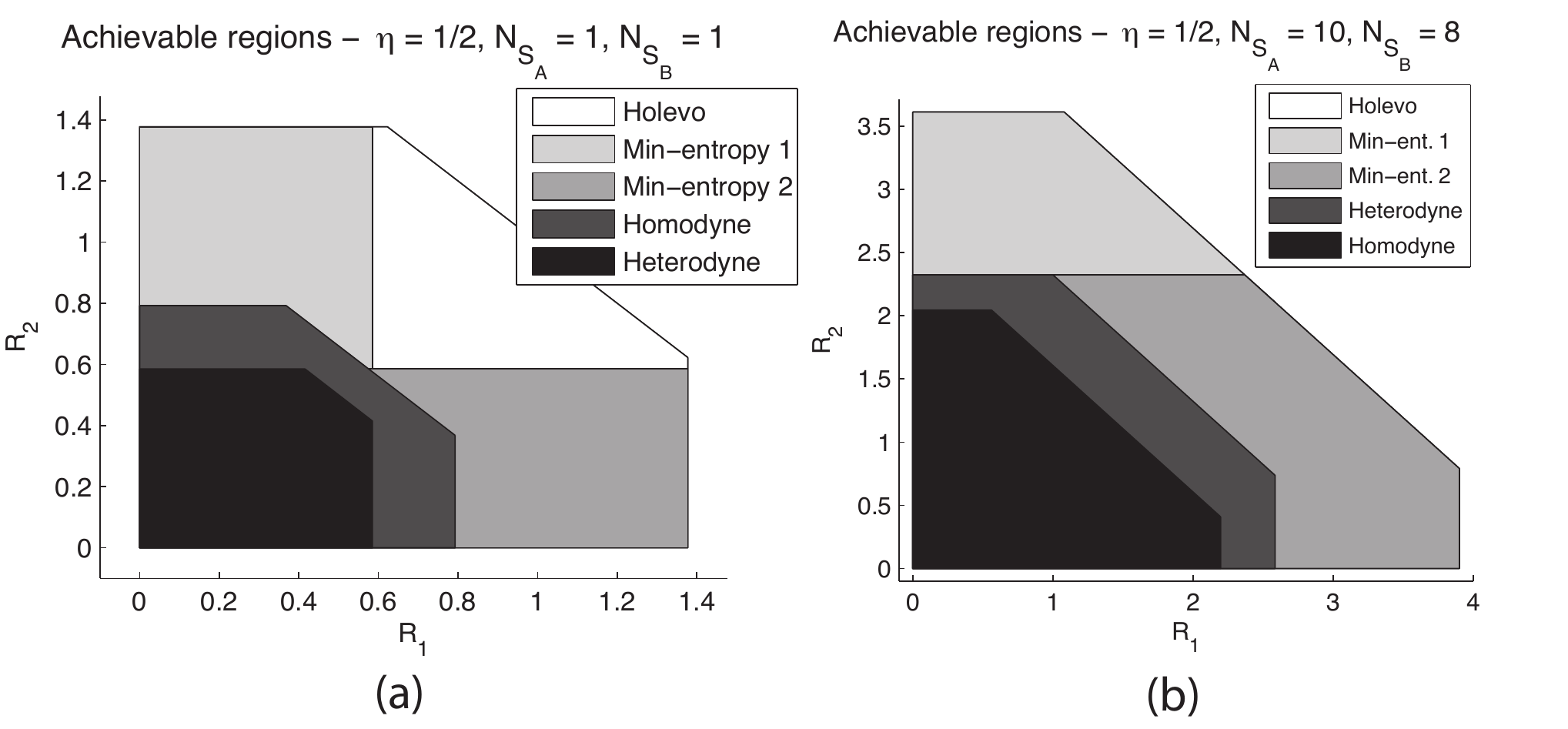}%
\caption{Rate regions achievable with various strategies at the receiving end
such as Holevo joint detection (with an unrealistic square-root measurement),
sequential decoding with a \textquotedblleft vacuum-or-not\textquotedblright%
\ measurement according to Corollary~\ref{rem:full-region-convex-hull},
heterodyne detection, or homodyne detection. The convex hull of the regions
entitled \textquotedblleft Min-entropy 1\textquotedblright\ and
\textquotedblleft Min-entropy 2\textquotedblright\ is the region given by
Corollary~\ref{rem:full-region-convex-hull}.}%
\label{fig:regions}%
\end{center}
\end{figure}

Figure~\ref{fig:equal-regions}\ more generally captures the values of the mean
input photon numbers $N_{S_{A}}$ and $N_{S_{B}}$ such that the region from
Corollary~\ref{rem:full-region-convex-hull}\ is equal to the Yen-Shapiro
region. The figure plots these values for two fixed values of the transmissivity:
$\eta=1/2$ and $\eta=4/5$. A simple observation is that the regions are
equivalent for higher mean input photon numbers. This result follows because
the min-entropy boundaries are equivalent to the some of the boundaries in the
heterodyne detection region (c.f., Ref.~\cite{YS05}), and we know that heterodyne
detection becomes optimal in the high photon-number limit.%
\begin{figure}
[ptb]
\begin{center}
\includegraphics[
natheight=2.892800in,
natwidth=6.447200in,
height=1.5592in,
width=3.4411in
]%
{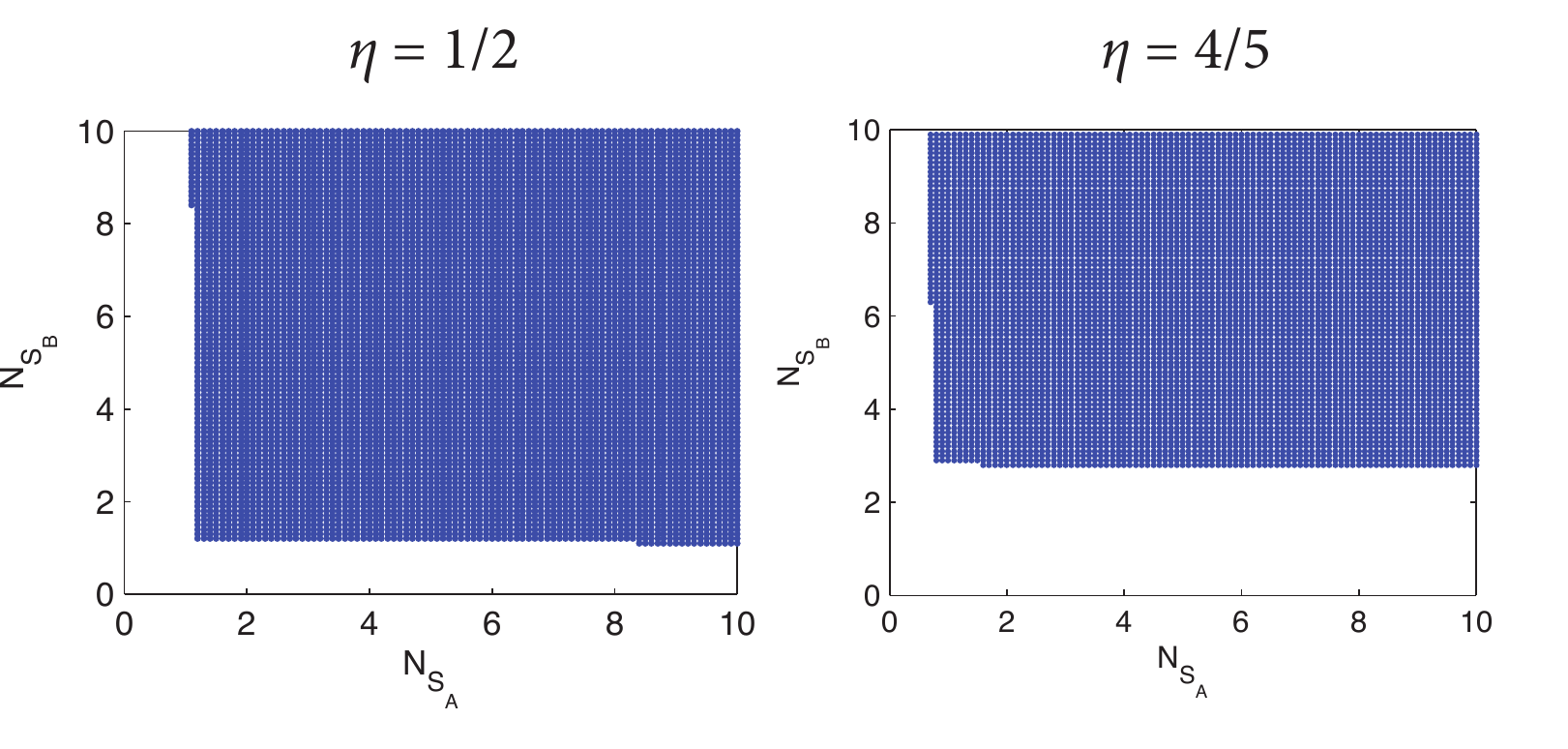}%
\caption{ The shaded region indicates the values of the mean input photon
numbers $N_{S_{A}}$ and $N_{S_{B}}$ such that the region from
Corollary~\ref{rem:full-region-convex-hull}\ is equal to the Yen-Shapiro
region.}%
\label{fig:equal-regions}%
\end{center}
\end{figure}

\section{Conclusion}
\label{sec:conc}
We have provided a near-explicit physically-realizable optical receiver such that two senders and a
receiver can achieve, in some cases, the Yen-Shapiro rate region in
(\ref{eq:YS-region}) for the pure-interference bosonic multiple access
channel. The scheme has the receiver perform binary-outcome quantum tests for
every codeword pair in the two codebooks of the senders. It is possible to
implement this strategy with unitary displacements (realizable
using a bank of highly-transmissive beamsplitters and strong
coherent-state local oscillators), and a \textquotedblleft
vacuum-or-not\textquotedblright\ measurement (for which an all-optical
structured realization still eludes us). However,
there is a suggestion for realizing this measurement using an atom-optical
coupled system with adiabatic STIRAP pulses~\cite{Oi2012}.

There are many open questions to consider. First, it seems natural to
conjecture that the sequential decoding algorithm in
Section~\ref{sec:rec-alg-bos-MAC} should be able to achieve the full
Yen-Shapiro rate region. It is a bit odd that its performance should
depend on the particular mathematical error analysis employed, either that in
Theorem~\ref{thm:main-theorem}\ or Corollary~\ref{rem:full-region-convex-hull}%
, given that the physical procedure for decoding is the same in both cases. It
is very likely our error analysis that is lacking, and we think that an
eventual proof of the \textquotedblleft quantum simultaneous decoding
conjecture\textquotedblright\ from Ref.~\cite{FHSSW11}\ might resolve this issue.

Yen and Shapiro found that employing squeezed states for an encoding could
achieve rates beyond the coherent-state region in (\ref{eq:YS-region}) for the
same values of $\eta$, $N_{S_{A}}$, and $N_{S_{B}}$ \cite{YS05}. It would be
interesting to develop a physically-realizable sequential decoding strategy
for this case. Though, it is not clear to us how to do so because the state at
the output of the channel in (\ref{eq:bosonic-MAC}) in this case will be a
mixed state, and as of now, we do not know how to realize such a decoder with
optical devices.

One can also achieve the rate region of the multiple access channel with a
successive decoder \cite{W01}, in which the receiver first decodes one
sender's message before decoding the other sender's message. It is of course
possible to do this in principle, but it is not clear to us how to implement
such a decoder with optical devices. If we knew how to realize a sequential
decoder for the channel in (\ref{eq:pure-loss-bos}) where the environment
injects thermal noise, then this should lead to an implementation of a
successive decoder for the pure-interference bosonic multiple access channel.

Finally, it is open to find a physically-realizable sequential decoding scheme
for an entanglement-assisted bosonic multiple access channel, for which an
achievable rate region was given in Ref.~\cite{XW11}. It is also open to
determine a physically-realizable decoder for the bosonic multiple access
channel with thermal noise \cite{YS05}, the bosonic broadcast channel
\cite{GS07,GSE07}, and the bosonic quantum interference channel \cite{GSW11}.
The authors thank Ivan Savov for a helpful discussion.

\bibliographystyle{IEEEtran}
\bibliography{Ref}

\appendices

\section{Typical Sequences and Typical Subspaces}

\label{sec:typ-review}Consider a density operator $\rho$ with the following
spectral decomposition:%
\[
\rho=\sum_{x}p_{X}\left(  x\right)  \left\vert x\right\rangle \left\langle
x\right\vert .
\]
The weakly typical subspace is defined as the span of all vectors such that
the sample entropy $\overline{H}\left(  x^{n}\right)  $ of their classical
label is close to the true entropy $H\left(  X\right)  $ of the distribution
$p_{X}\left(  x\right)  $ \cite{book2000mikeandike,W11}:%
\[
T_{\delta}^{X^{n}}\equiv\text{span}\left\{  \left\vert x^{n}\right\rangle
:\left\vert \overline{H}\left(  x^{n}\right)  -H\left(  X\right)  \right\vert
\leq\delta\right\}  ,
\]
where%
\begin{align*}
\overline{H}\left(  x^{n}\right)   &  \equiv-\frac{1}{n}\log\left(  p_{X^{n}%
}\left(  x^{n}\right)  \right)  ,\\
H\left(  X\right)   &  \equiv-\sum_{x}p_{X}\left(  x\right)  \log p_{X}\left(
x\right)  .
\end{align*}
The projector $\Pi_{\rho,\delta}^{n}$\ onto the typical subspace of $\rho$ is
defined as%
\[
\Pi_{\rho,\delta}^{n}\equiv\sum_{x^{n}\in T_{\delta}^{X^{n}}}\left\vert
x^{n}\right\rangle \left\langle x^{n}\right\vert ,
\]
where we have \textquotedblleft overloaded\textquotedblright\ the symbol
$T_{\delta}^{X^{n}}$ to refer also to the set of $\delta$-typical sequences:%
\[
T_{\delta}^{X^{n}}\equiv\left\{  x^{n}:\left\vert \overline{H}\left(
x^{n}\right)  -H\left(  X\right)  \right\vert \leq\delta\right\}  .
\]
The three important properties of the typical projector are as follows:%
\begin{align}
\text{Tr}\left\{  \Pi_{\rho,\delta}^{n}\rho^{\otimes n}\right\}   &
\geq1-\epsilon,\label{eq:prop-1-typ}\\
\text{Tr}\left\{  \Pi_{\rho,\delta}^{n}\right\}   &  \leq2^{n\left[  H\left(
X\right)  +\delta\right]  },\nonumber\\
2^{-n\left[  H\left(  X\right)  +\delta\right]  }\ \Pi_{\rho,\delta}^{n}  &
\leq\Pi_{\rho,\delta}^{n}\ \rho^{\otimes n}\ \Pi_{\rho,\delta}^{n}%
\leq2^{-n\left[  H\left(  X\right)  -\delta\right]  }\ \Pi_{\rho,\delta}^{n},
\label{eq:typ-prop-three}%
\end{align}
where the first property holds for arbitrary $\epsilon,\delta>0$ and
sufficiently large $n$. Consider an ensemble $\left\{  p_{X}\left(  x\right)
,\rho_{x}\right\}  _{x\in\mathcal{X}}$ of states. Suppose that each state
$\rho_{x}$ has the following spectral decomposition:%
\[
\rho_{x}=\sum_{y}p_{Y|X}\left(  y|x\right)  \left\vert y_{x}\right\rangle
\left\langle y_{x}\right\vert .
\]
Consider a density operator $\rho_{x^{n}}$ which is conditional on a classical
sequence $x^{n}\equiv x_{1}\cdots x_{n}$:%
\[
\rho_{x^{n}}\equiv\rho_{x_{1}}\otimes\cdots\otimes\rho_{x_{n}}.
\]
We define the weak conditionally typical subspace as the span of vectors
(conditional on the sequence $x^{n}$) such that the sample conditional entropy
$\overline{H}\left(  y^{n}|x^{n}\right)  $ of their classical labels is close
to the true conditional entropy $H\left(  Y|X\right)  $ of the distribution
$p_{Y|X}\left(  y|x\right)  p_{X}\left(  x\right)  $
\cite{book2000mikeandike,W11}:%
\[
T_{\delta}^{Y^{n}|x^{n}}\equiv\text{span}\left\{  \left\vert y_{x^{n}}%
^{n}\right\rangle :\left\vert \overline{H}\left(  y^{n}|x^{n}\right)
-H\left(  Y|X\right)  \right\vert \leq\delta\right\}  ,
\]
where%
\begin{align*}
\overline{H}\left(  y^{n}|x^{n}\right)   &  \equiv-\frac{1}{n}\log\left(
p_{Y^{n}|X^{n}}\left(  y^{n}|x^{n}\right)  \right)  ,\\
H\left(  Y|X\right)   &  \equiv-\sum_{x}p_{X}\left(  x\right)  \sum_{y}%
p_{Y|X}\left(  y|x\right)  \log p_{Y|X}\left(  y|x\right)  .
\end{align*}
The projector $\Pi_{\rho_{x^{n}},\delta}$ onto the weak conditionally typical
subspace of $\rho_{x^{n}}$ is as follows:%
\[
\Pi_{\rho_{x^{n}},\delta}\equiv\sum_{y^{n}\in T_{\delta}^{Y^{n}|x^{n}}%
}\left\vert y_{x^{n}}^{n}\right\rangle \left\langle y_{x^{n}}^{n}\right\vert
,
\]
where we have again overloaded the symbol $T_{\delta}^{Y^{n}|x^{n}}$ to refer
to the set of weak conditionally typical sequences:%
\[
T_{\delta}^{Y^{n}|x^{n}}\equiv\left\{  y^{n}:\left\vert \overline{H}\left(
y^{n}|x^{n}\right)  -H\left(  Y|X\right)  \right\vert \leq\delta\right\}  .
\]
The three important properties of the weak conditionally typical projector are
as follows:%
\begin{align}
\mathbb{E}_{X^{n}}\left\{  \text{Tr}\left\{  \Pi_{\rho_{X^{n}},\delta}%
\rho_{X^{n}}\right\}  \right\}   &  \geq1-\epsilon,\label{eq:prop-1-weak-typ}\\
\text{Tr}\left\{  \Pi_{\rho_{x^{n}},\delta}\right\}   &  \leq2^{n\left[
H\left(  Y|X\right)  +\delta\right]  },\label{eq:property-2-typical}\\
2^{-n\left[  H\left(  Y|X\right)  +\delta\right]  }\ \Pi_{\rho_{x^{n}}%
,\delta}  &  \leq\Pi_{\rho_{x^{n}},\delta}\ \rho_{x^{n}}\ \Pi_{\rho_{x^{n}%
},\delta}\label{eq:property-3-typical}\\
&  \leq2^{-n\left[  H\left(  Y|X\right)  -\delta\right]  }\ \Pi_{\rho_{x^{n}%
},\delta},\nonumber
\end{align}
where the first property holds for arbitrary $\epsilon,\delta>0$ and
sufficiently large $n$, and the expectation is with respect to the
distribution $p_{X^{n}}\left(  x^{n}\right)  $.

\section{Useful lemmas}

\label{sec:useful-lemmas}Here we collect some useful lemmas.

\begin{lemma}
[Gentle Operator Lemma \cite{itit1999winter,ON07}]\label{lem:gentle-operator}%
Let $\Lambda$ be a positive operator where $0\leq\Lambda\leq I$ (usually
$\Lambda$ is a POVM\ element), $\rho$ a state, and $\epsilon$ a positive
number such that the probability of detecting the outcome $\Lambda$ is high:%
\[
\text{Tr}\left\{  \Lambda\rho\right\}  \geq1-\epsilon.
\]
Then the measurement causes little disturbance to the state $\rho$:%
\[
\left\Vert \rho-\sqrt{\Lambda}\rho\sqrt{\Lambda}\right\Vert _{1}\leq
2\sqrt{\epsilon}.
\]

\end{lemma}

The following lemma appears in Refs.~\cite{itit1999winter,ON07,W11}.

\begin{lemma}
[Gentle Operator Lemma for Ensembles]\label{lem:gentle-operator-ens} Given an
ensemble $\left\{  p_{X}\left(  x\right)  ,\rho_{x}\right\}  $ with expected
density operator $\rho\equiv\sum_{x}p_{X}\left(  x\right)  \rho_{x}$, suppose
that an operator $\Lambda$ such that $I\geq\Lambda\geq0$ succeeds with high
probability on the state $\rho$:%
\[
\text{Tr}\left\{  \Lambda\rho\right\}  \geq1-\epsilon.
\]
Then the subnormalized state $\sqrt{\Lambda}\rho_{x}\sqrt{\Lambda}$ is close
in expected trace distance to the original state $\rho_{x}$:%
\[
\mathbb{E}_{X}\left\{  \left\Vert \sqrt{\Lambda}\rho_{X}\sqrt{\Lambda}%
-\rho_{X}\right\Vert _{1}\right\}  \leq2\sqrt{\epsilon}.
\]

\end{lemma}

\begin{lemma}
\label{lem:trace-inequality}Let $\rho$ and $\sigma$ be positive operators and
$\Lambda$ a positive operator such that $0\leq\Lambda\leq I$. Then the
following inequality holds%
\[
\text{Tr}\left\{  \Lambda\rho\right\}  \leq\text{Tr}\left\{  \Lambda
\sigma\right\}  +\left\Vert \rho-\sigma\right\Vert _{1}.
\]

\end{lemma}

\begin{lemma}
[Non-commutative union bound \cite{S11}]\label{lem-non-com-union-bound}Let
$\sigma$ be a subnormalized state such that $\sigma\geq0$ and Tr$\left\{
\sigma\right\}  \leq1$. Let $\Pi_{1}$, \ldots, $\Pi_{N}$ be projectors. Then
the following \textquotedblleft non-commutative union bound\textquotedblright%
\ holds%
\[
\text{Tr}\left\{  \sigma\right\}  -\text{Tr}\left\{  \Pi_{N}\cdots\Pi
_{1}\sigma\Pi_{1}\cdots\Pi_{N}\right\}  \leq2\sqrt{\sum_{i=1}^{N}%
\text{Tr}\left\{  \left(  I-\Pi_{i}\right)  \sigma\right\}  }.
\]

\end{lemma}

\section{Proof of main theorem}
\label{sec:err-analysis}
\textbf{Error Analysis.} Suppose that Sender~1 transmits the $l^{\text{th}}$
codeword and that Sender~2 transmits the $m^{\text{th}}$ codeword. Then the
probability for the receiver to decode correctly with the above sequential
decoding strategy is as follows:%
\[
\text{Tr}\left\{  \phi_{l,m}\hat{\Pi}_{l-1,m}\cdots\hat{\Pi}_{1,1}\phi
_{l,m}\hat{\Pi}_{1,1}\cdots\hat{\Pi}_{l-1,m}\phi_{l,m}\right\}  ,
\]
where we make the abbreviations%
\begin{align*}
\phi_{l,m}  &  \equiv\left\vert \phi_{x^{n}\left(  l\right)  ,y^{n}\left(
m\right)  }\right\rangle \left\langle \phi_{x^{n}\left(  l\right)
,y^{n}\left(  m\right)  }\right\vert ,\\
\hat{\Pi}_{l,m}  &  \equiv I-\phi_{l,m}.
\end{align*}
The above probability corresponds to the case that the receiver receives
\textquotedblleft no\textquotedblright\ answers when he performs the
measurements for the 1st codeword pair $\left(  x^{n}\left(  1\right)
,y^{n}\left(  1\right)  \right)  $ all the way until the codeword pair
$\left(  x^{n}\left(  l-1\right)  ,y^{n}\left(  m\right)  \right)  $ and he
then receives a \textquotedblleft yes\textquotedblright\ answer for the
codeword pair $\left(  x^{n}\left(  l\right)  ,y^{n}\left(  m\right)  \right)
$. So, the probability that the receiver decodes the pair $\left(  l,m\right)
$ incorrectly is%
\[
1-\text{Tr}\left\{  \phi_{l,m}\hat{\Pi}_{l-1,m}\cdots\hat{\Pi}_{1,1}\phi
_{l,m}\hat{\Pi}_{1,1}\cdots\hat{\Pi}_{l-1,m}\phi_{l,m}\right\}  .
\]
In order to simplify the error analysis, we analyze the expectation of the
above error probability, by assuming that both senders choose their messages
uniformly at random and furthermore that the codewords are selected at random
independently and identically according to the distributions $p_{X}\left(
x\right)  $ and $p_{Y}\left(  y\right)  $ (as described above):%
\begin{equation}
1-\mathbb{E}\text{Tr}\left\{  \phi_{L,M}\hat{\Pi}_{L-1,M}\cdots\hat{\Pi}%
_{1,1}\phi_{L,M}\hat{\Pi}_{1,1}\cdots\hat{\Pi}_{L-1,M}\phi_{L,M}\right\}  .
\label{eq:seq-dec-err-prob}%
\end{equation}
In the above and for the rest of the proof, it is implicit that the
expectation $\mathbb{E}$ is with respect to the random variables $X^{n}$,
$Y^{n}$, $L$, and $M$, unless otherwise stated.

We first observe that it is possible to consider a slightly altered channel
for which we are coding. Instead of decoding the original channel
$\left\vert \phi_{x^{n}\left(  l\right)  ,y^{n}\left(
m\right)  }\right\rangle$, we can
decode a projected channel of the form $\Pi\Pi_{l} \left\vert \phi_{x^{n}\left(  l\right)  ,y^{n}\left(
m\right)  }\right\rangle$, where we define the projectors $\Pi$ and $\Pi_l$ below. That we can do so follows from the inequalities below:
\begin{align*}
1  &  =\mathbb{E}\text{Tr}\left\{  \phi_{L,M}\right\} \\
&  =\mathbb{E}\text{Tr}\left\{  \Pi_{L}\phi_{L,M}\right\}  +\mathbb{E}%
\text{Tr}\{\hat{\Pi}_{L}\phi_{L,M}\}\\
&  =\mathbb{E}\text{Tr}\left\{  \Pi_{L}\phi_{L,M}\Pi_{L}\right\}
+\mathbb{E}_{X^{n},L,M}\text{Tr}\{\hat{\Pi}_{L}\mathbb{E}_{Y^{n}}\left\{
\phi_{L,M}\right\}  \}\\
&  =\mathbb{E}\text{Tr}\left\{  \Pi_{L}\phi_{L,M}\Pi_{L}\right\}
+\mathbb{E}_{X^{n},L,M}\text{Tr}\{\hat{\Pi}_{L}\rho_{L}\}\\
&  \leq\mathbb{E}\text{Tr}\left\{  \Pi\Pi_{L}\phi_{L,M}\Pi_{L}\right\}
+\mathbb{E}\text{Tr}\left\{  \hat{\Pi}\Pi_{L}\phi_{L,M}\Pi_{L}\right\}
+\mathbb{\epsilon}
\end{align*}
\begin{align*}
&  \leq\mathbb{E}\text{Tr}\left\{  \Pi\Pi_{L}\phi_{L,M}\Pi_{L}\Pi\right\}
+\mathbb{E}\text{Tr}\left\{  \hat{\Pi}\phi_{L,M}\right\} \\
&  \ \ \ \ +\mathbb{E}\left\Vert \phi_{L,M}-\Pi_{L}\phi_{L,M}\Pi
_{L}\right\Vert _{1}+\mathbb{\epsilon}\\
&  \leq\mathbb{E}\text{Tr}\left\{  \Pi\Pi_{L}\phi_{L,M}\Pi_{L}\Pi\right\}
+\text{Tr}\left\{  \hat{\Pi}\mathbb{E}\left\{  \phi_{L,M}\right\}  \right\}
+2\sqrt{\mathbb{\epsilon}}+\mathbb{\epsilon}\\
&  =\mathbb{E}\text{Tr}\left\{  \Pi\Pi_{L}\phi_{L,M}\Pi_{L}\Pi\right\}
+\text{Tr}\left\{  \hat{\Pi}\rho^{\otimes n}\right\}  +2\sqrt{\mathbb{\epsilon
}}+\mathbb{\epsilon}\\
&  \leq\mathbb{E}\text{Tr}\left\{  \Pi\Pi_{L}\phi_{L,M}\Pi_{L}\Pi\right\}
+2\sqrt{\mathbb{\epsilon}}+2\mathbb{\epsilon}%
\end{align*}
In the second equality, the projector $\Pi_{L}$ is a weak conditionally
typical projector corresponding to the following state:%
\[
\mathbb{E}_{Y^{n}}\left\{  \left\vert \phi_{X^{n}\left(  L\right)
,Y^{n}\left(  M\right)  }\right\rangle \left\langle \phi_{X^{n}\left(
L\right)  ,Y^{n}\left(  M\right)  }\right\vert \right\}  .
\]
(See Appendix~\ref{sec:typ-review} for an explanation.) The first inequality follows from
applying the property (\ref{eq:prop-1-weak-typ}) of weak conditionally typical subspaces. Also, we
bring in the weak typical projector (defined as $\Pi$)\ for the following
state:%
\[
\mathbb{E}_{X^{n}Y^{n}}\left\{  \left\vert \phi_{X^{n}\left(  L\right)
,Y^{n}\left(  M\right)  }\right\rangle \left\langle \phi_{X^{n}\left(
L\right)  ,Y^{n}\left(  M\right)  }\right\vert \right\}  .
\]
The second inequality follows by applying the trace inequality from Lemma~\ref{lem:trace-inequality}.
The third inequality follows from the Gentle Operator Lemma for Ensembles and
the property (\ref{eq:prop-1-weak-typ}) of weak conditionally typical subspaces. The final inequality
follows from the property (\ref{eq:prop-1-typ}) of weakly typical subspaces.

Consider also the following lower bound:%
\begin{align*}
&  \mathbb{E}\text{Tr}\left\{  \phi_{L,M}\hat{\Pi}_{L-1,M}\cdots\hat{\Pi
}_{1,1}\phi_{L,M}\hat{\Pi}_{1,1}\cdots\hat{\Pi}_{L-1,M}\phi_{L,M}\right\} \\
&  =\mathbb{E}\text{Tr}\left\{  \hat{\Pi}_{1,1}\cdots\hat{\Pi}_{L-1,M}%
\phi_{L,M}\hat{\Pi}_{L-1,M}\cdots\hat{\Pi}_{1,1}\phi_{L,M}\right\} \\
&  \geq\mathbb{E}\text{Tr}\{\hat{\Pi}_{1,1}\cdots\hat{\Pi}_{L-1,M}\phi
_{L,M}\hat{\Pi}_{L-1,M}\cdots\hat{\Pi}_{1,1}\Pi\Pi_{L}\phi_{L,M}\Pi_{L}\Pi\}\\
&  \ \ \ \ -\mathbb{E}\left\Vert \Pi_{L}\phi_{L,M}\Pi_{L}-\phi_{L,M}%
\right\Vert _{1}\\
&  \ \ \ \ -\mathbb{E}\left\Vert \Pi\phi_{L,M}\Pi-\phi_{L,M}\right\Vert _{1}\\
&  \geq\mathbb{E}\text{Tr}\left\{  \hat{\Pi}_{1,1}\cdots\hat{\Pi}_{L-1,M}%
\phi_{L,M}\hat{\Pi}_{L-1,M}\cdots\hat{\Pi}_{1,1}\Pi\Pi_{L}\phi_{L,M}\Pi_{L}%
\Pi\right\} \\
&  \ \ \ \ -4\sqrt{\epsilon}.
\end{align*}
The first equality is from cyclicity of trace. The first inequality follows
from two applications of the trace inequality (Lemma~\ref{lem:trace-inequality}). The final inequality follows
from the properties of typical subspaces. Putting all of this together gives
us the following upper bound on the error probability in
(\ref{eq:seq-dec-err-prob}):%
\begin{multline*}
\mathbb{E}\text{Tr}\left\{  \Pi\Pi_{L}\phi_{L,M}\Pi_{L}\Pi\right\} \\
-\mathbb{E}\text{Tr}\{\phi_{L,M}\hat{\Pi}_{L-1,M}\cdots\hat{\Pi}_{1,1}\Pi
\Pi_{L}\phi_{L,M}\Pi_{L}\Pi\hat{\Pi}_{1,1}\cdots\hat{\Pi}_{L-1,M}\}\\
+2\mathbb{\epsilon}+6\sqrt{\epsilon}.
\end{multline*}
We now apply Sen's non-commutative union bound (Lemma~\ of Ref.~\cite{S11}\ or
Lemma~\ref{lem-non-com-union-bound} of Appendix~\ref{sec:useful-lemmas}) and concavity of the
square-root function to obtain the following upper bound on the error
probability:%
\begin{multline*}
2\left(
\begin{array}
[c]{l}%
\mathbb{E}\text{Tr}\left\{  \left(  I-\phi_{L,M}\right)  \Pi\Pi_{L}\phi
_{L,M}\Pi_{L}\Pi\right\} \\
\ \ \ \ \ +\mathbb{E}\sum_{\left(  i,j\right)  <\left(  L,M\right)  }%
\text{Tr}\left\{  \phi_{i,j}\Pi\Pi_{L}\phi_{L,M}\Pi_{L}\Pi\right\}
\end{array}
\right)  ^{\frac{1}{2}}\\
\leq2\left(
\begin{array}
[c]{l}%
\mathbb{E}\text{Tr}\left\{  \left(  I-\phi_{L,M}\right)  \Pi\Pi_{L}\phi
_{L,M}\Pi_{L}\Pi\right\} \\
\ \ \ \ \ +\mathbb{E}\sum_{\left(  i,j\right)  \neq\left(  L,M\right)
}\text{Tr}\left\{  \phi_{i,j}\Pi\Pi_{L}\phi_{L,M}\Pi_{L}\Pi\right\}
\end{array}
\right)  ^{\frac{1}{2}}.
\end{multline*}
We handle each of these error terms individually. We upper bound the first
term:%
\begin{align*}
&  \mathbb{E}\text{Tr}\left\{  \left(  I-\phi_{L,M}\right)  \Pi\Pi_{L}%
\phi_{L,M}\Pi_{L}\Pi\right\} \\
&  \leq\mathbb{E}\text{Tr}\left\{  \left(  I-\phi_{L,M}\right)  \phi
_{L,M}\right\} \\
&  \ \ \ \ +\mathbb{E}\left\Vert \Pi_{L}\phi_{L,M}\Pi_{L}-\phi_{L,M}%
\right\Vert _{1}+\mathbb{E}\left\Vert \Pi\phi_{L,M}\Pi-\phi_{L,M}\right\Vert
_{1}\\
&  \leq4\sqrt{\epsilon}.
\end{align*}
The inequalities follow from the trace inequality, the properties of typical
subspaces, and the Gentle Operator Lemma for Ensembles. We can split the
second term into three different ones as follows:%
\[
\mathbb{E}\sum_{\left(  i,j\right)  \neq\left(  L,M\right)  }\left(
\cdot\right)  =\mathbb{E}\sum_{i\neq L}\left(  \cdot\right)  +\mathbb{E}%
\sum_{j\neq M}\left(  \cdot\right)  +\mathbb{E}\sum_{i\neq L,j\neq M}\left(
\cdot\right)  .
\]
We handle each of these three terms separately. Consider the first term:%
\begin{align}
&  \mathbb{E}\sum_{i\neq L}\text{Tr}\left\{  \phi_{i,M}\Pi\Pi_{L}\phi_{L,M}%
\Pi_{L}\Pi\right\} \nonumber\\
&  =\mathbb{E}_{Y^{n},L,M}\sum_{i\neq L}\mathbb{E}_{X^{n}}\text{Tr}\left\{
\phi_{i,M}\Pi\Pi_{L}\phi_{L,M}\Pi_{L}\Pi\right\} \nonumber\\
&  =\mathbb{E}_{Y^{n},L,M}\sum_{i\neq L}\text{Tr}\left\{  \mathbb{E}_{X^{n}%
}\left\{  \phi_{i,M}\right\}  \Pi\Pi_{L}\mathbb{E}_{X^{n}}\left\{  \phi
_{L,M}\right\}  \Pi_{L}\Pi\right\} \nonumber\\
&  =\mathbb{E}_{Y^{n},L,M}\sum_{i\neq L}\text{Tr}\left\{  \rho_{M}\Pi\Pi
_{L}\rho_{M}\Pi_{L}\Pi\right\} \nonumber\\
&  \leq2^{-nH_{\min}\left(  B|Y\right)  }\mathbb{E}_{Y^{n},L,M}\sum_{i\neq
L}\text{Tr}\left\{  \rho_{M}\Pi\Pi_{L}\Pi\right\} \nonumber\\
&  \leq2^{-nH_{\min}\left(  B|Y\right)  }\left\vert \mathcal{L}\right\vert
\label{eq:min-ent-bound}%
\end{align}
The first equality follows by bringing the expectation $\mathbb{E}_{X^{n}}$
inside the sum. The second equality follows because the random variables
$X^{n}\left(  i\right)  $ and $X^{n}\left(  L\right)  $ are independent and so
the expectation $\mathbb{E}_{X^{n}}$ distributes. The third equality follows
by evaluating the expectations by defining $\rho_{M}$ to be as follows:%
\[
\rho_{M}\equiv\mathbb{E}_{X^{n}}\left\{  \left\vert \phi_{X^{n}\left(
L\right)  ,Y^{n}\left(  M\right)  }\right\rangle \left\langle \phi
_{X^{n}\left(  L\right)  ,Y^{n}\left(  M\right)  }\right\vert \right\}  .
\]
The first inequality follows by bounding the largest eigenvalue of $\rho_{M}$
by the min-entropy $2^{-nH_{\min}\left(  B|Y\right)  }$. The second inequality
follows because Tr$\left\{  \rho_{M}\Pi\Pi_{L}\Pi\right\}  \leq1$.

We handle the second term:%
\begin{align*}
&  \mathbb{E}\sum_{j\neq M}\text{Tr}\left\{  \phi_{L,j}\Pi\Pi_{L}\phi_{L,M}%
\Pi_{L}\Pi\right\} \\
&  =\mathbb{E}_{X^{n},L,M}\sum_{j\neq M}\mathbb{E}_{Y^{n}}\text{Tr}\left\{
\phi_{L,j}\Pi\Pi_{L}\phi_{L,M}\Pi_{L}\Pi\right\} \\
&  =\mathbb{E}_{X^{n},L,M}\sum_{j\neq M}\text{Tr}\left\{  \mathbb{E}_{Y^{n}%
}\left\{  \phi_{L,j}\right\}  \Pi\Pi_{L}\mathbb{E}_{Y^{n}}\left\{  \phi
_{L,M}\right\}  \Pi_{L}\Pi\right\} \\
&  =\mathbb{E}_{X^{n},L,M}\sum_{j\neq M}\text{Tr}\left\{  \rho_{L}\Pi\Pi
_{L}\rho_{L}\Pi_{L}\Pi\right\} \\
&  \leq2^{-n\left[  H\left(  B|X\right)  -\delta\right]  }\mathbb{E}%
_{X^{n},L,M}\sum_{j\neq M}\text{Tr}\left\{  \rho_{L}\Pi\Pi_{L}\Pi\right\} \\
&  \leq2^{-n\left[  H\left(  B|X\right)  -\delta\right]  }\left\vert
\mathcal{M}\right\vert .
\end{align*}
The first four equalities follow for reasons similar to the above. The first
inequality is from the typical projector bound in (\ref{eq:property-3-typical}). The last inequality
follows because Tr$\left\{  \rho_{L}\Pi\Pi_{L}\Pi\right\}  \leq1$. Finally, we
handle the third term:%
\begin{align*}
&  \mathbb{E}\sum_{i\neq L,j\neq M}\text{Tr}\left\{  \phi_{i,j}\Pi\Pi_{L}%
\phi_{L,M}\Pi_{L}\Pi\right\} \\
&  =\mathbb{E}_{X^{n},L,M}\sum_{i\neq L,j\neq M}\text{Tr}\left\{  \rho_{i}%
\Pi\Pi_{L}\rho_{L}\Pi_{L}\Pi\right\} \\
&  \leq\mathbb{E}_{X^{n},L,M}\sum_{i\neq L,j\neq M}\text{Tr}\left\{  \rho
_{i}\Pi\rho_{L}\Pi\right\}
\end{align*}%
\begin{align*}
&  =\mathbb{E}_{L,M}\sum_{i\neq L,j\neq M}\text{Tr}\left\{  \mathbb{E}_{X^{n}%
}\left\{  \rho_{i}\right\}  \Pi\mathbb{E}_{X^{n}}\left\{  \rho_{L}\right\}
\Pi\right\} \\
&  =\mathbb{E}_{L,M}\sum_{i\neq L,j\neq M}\text{Tr}\left\{  \rho^{\otimes
n}\Pi\rho^{\otimes n}\Pi\right\} \\
&  \leq2^{-n\left[  H\left(  B\right)  -\delta\right]  }\mathbb{E}_{L,M}%
\sum_{i\neq L,j\neq M}\text{Tr}\left\{  \rho^{\otimes n}\Pi\right\} \\
&  \leq2^{-n\left[  H\left(  B\right)  -\delta\right]  }\left\vert
\mathcal{L}\right\vert \left\vert \mathcal{M}\right\vert .
\end{align*}
The first inequality follows for reasons similar to the above ones. The first
inequality follows because $\Pi_{L}\rho_{L}\Pi_{L}\leq\rho_{L}$. The second
equality follows from distributing the expectation over $X^{n}$. The second
inequality follows from the typical projector bound in (\ref{eq:typ-prop-three}). The final
inequality follows because Tr$\left\{  \rho^{\otimes n}\Pi\right\}  \leq1$.

Thus, the overall upper bound on the error probability with this sequential
decoding strategy is%
\begin{multline}
6 \epsilon+2\sqrt{\epsilon} + 
2 \big( 4\sqrt{\epsilon} + 2^{-nH_{\min}\left(  B|Y\right)  }\left\vert \mathcal{L}\right\vert \\
+ 2^{-n\left[  H\left(  B|X\right)  -\delta\right]  }\left\vert
\mathcal{M}\right\vert + 2^{-n\left[  H\left(  B\right)  -\delta\right]  }\left\vert
\mathcal{L}\right\vert \left\vert \mathcal{M}\right\vert\big)^{1/2} ,
\end{multline}
which we can make arbitrarily small by choosing the rates to be in the region given
in the statement of Theorem~\ref{thm:main-theorem} and taking $n$
sufficiently large. We proved a bound on the
expectation of the average probability, which implies there exists a
particular code that has arbitrarily small average error probability under the
same choice of $\left\vert \mathcal{L}\right\vert $, $\left\vert \mathcal{M}\right\vert $, and $n$. 

\begin{remark}
\label{rem:connection-to-Sen}
One can achieve the following rate region with von Neumann entropies by
employing an idea similar to that of Sen in Ref.~\cite{S11}:%
\begin{equation}
R_{1}\leq H\left(  B|Y\right)  ,\ \ \ R_{2}\leq H\left(  B|X\right)
,\ \ \ R_{1}+R_{2}\leq H\left(  B\right)  .
\label{eq:full-acheive-rate-region}%
\end{equation}
Indeed, the idea is to perform sequential decoding measurements of the
following form:%
\begin{equation}
\{\phi_{x^{n}\left(  l\right)  ,y^{n}\left(  m\right)  }^{\prime},I^{\otimes
n}-\phi_{x^{n}\left(  l\right)  ,y^{n}\left(  m\right)  }^{\prime}\},
\label{eq:mod-seq-meas}%
\end{equation}
where%
\[
|\phi_{x^{n}\left(  l\right)  ,y^{n}\left(  m\right)  }^{\prime}\rangle
\equiv\frac{1}{\left\Vert \Pi_{m}\left\vert \phi_{x^{n}\left(  l\right)
,y^{n}\left(  m\right)  }\right\rangle \right\Vert _{2}}\Pi_{m}\left\vert
\phi_{x^{n}\left(  l\right)  ,y^{n}\left(  m\right)  }\right\rangle ,
\]
and $\Pi_{m}$ is a typical projector for the state $\mathbb{E}_{X^{n}}\left\{
\phi_{X^{n}\left(  l\right)  ,y^{n}\left(  m\right)  }\right\}  $. The
following bound holds for the norm $\left\Vert \Pi_{m}\left\vert \phi
_{x^{n}\left(  l\right)  ,y^{n}\left(  m\right)  }\right\rangle \right\Vert
_{2}^{2}\geq1-\sqrt{\epsilon}$, due to the properties of quantum typicality
(one requires strong typicality here, but this is a minor point). After
applying Sen's bound, one can invoke the following operator inequality:%
\begin{align*}
&  \phi_{x^{n}\left(  l\right)  ,y^{n}\left(  m\right)  }^{\prime}\\
&  =\left(  \left\Vert \Pi_{m}\left\vert \phi_{x^{n}\left(  l\right)
,y^{n}\left(  m\right)  }\right\rangle \right\Vert _{2}^{2}\right)  ^{-1}%
\Pi_{m}\ \phi_{x^{n}\left(  l\right)  ,y^{n}\left(  m\right)  }\ \Pi_{m}\\
&  \leq\left(  1-\sqrt{\epsilon}\right)  ^{-1}\Pi_{m}\ \phi_{x^{n}\left(
l\right)  ,y^{n}\left(  m\right)  }\ \Pi_{m},
\end{align*}
and then employ typical subspace bounds in order to obtain a von Neumann
entropy bound rather than a min-entropy bound as in (\ref{eq:min-ent-bound}).
The reason we do not employ the above approach is that it is not clear to us
how to implement the measurements in (\ref{eq:mod-seq-meas}) with optical
devices when we get to the case of the pure-interference bosonic
MAC.
\end{remark}

\end{document}